\documentclass[11pt]{article}
\usepackage[margin=1in]{geometry}
\usepackage{graphicx}
\usepackage{amssymb}
\usepackage{fancybox}
\usepackage{amsmath}
\usepackage{authblk}
\usepackage{color}
\usepackage{rotating}
\definecolor{darkgray}{rgb}{0.25,0.25,0.25}
\definecolor{darkred}{rgb}{0.89,0.10,0.11}
\definecolor{darkblue}{rgb}{0.12,0.39,0.62}
\usepackage{url}
\newtheorem{theorem}{Theorem}
\newtheorem{lemma}[theorem]{Lemma}

\newenvironment{proof}[1][Proof]{\begin{trivlist}
\item[\hskip \labelsep {\bfseries #1}]}{\end{trivlist}}

\newcommand{\qed}{\nobreak \ifvmode \relax \else
      \ifdim\lastskip<1.5em \hskip-\lastskip
      \hskip1.5em plus0em minus0.5em \fi \nobreak
      \vrule height0.75em width0.5em depth0.25em\fi}

\begin{document}

\title{Defensive complexity and the phylogenetic conservation of immune control }

\author{Erick Chastain}
\affil{Department of Computer Science, Rutgers University, New Brunswick, NJ}
\author{Rustom Antia}
\affil{Department of Biology, Emory University, Atlanta, GA}
\author{Carl T. Bergstrom}
\affil{Department of Biology, University of Washington, Seattle, WA 98195-1800}
\affil{Santa Fe Institute, Santa Fe, NM}

\maketitle
 
 \begin{abstract}
 One strategy for winning a coevolutionary struggle is to evolve rapidly. Most of the literature on host-pathogen coevolution focuses on this phenomenon, and looks for consequent evidence of coevolutionary arms races.  An alternative strategy, less often considered in the literature, is to deter rapid evolutionary change by the opponent. To study how this can be done, we construct an evolutionary game between a controller that must process information, and an adversary that can tamper with this information processing. In this game, a species can foil its antagonist by processing information in a way that is hard for the antagonist to manipulate. We show that the structure of the information processing system induces a fitness landscape on which the adversary population evolves. Complex processing logic can carve long, deep fitness valleys that slow adaptive evolution in the adversary population. We suggest that this type of defensive complexity on the part of the vertebrate adaptive immune system may be an important element of coevolutionary dynamics between pathogens and their vertebrate hosts. Furthermore, we cite evidence that the immune control logic is phylogenetically conserved in mammalian lineages. Thus our model of defensive complexity suggests a new hypothesis for the lower rates of evolution for immune control logic compared to other immune structures.  
\end{abstract}

\section{Introduction}

Coevolution is often antagonistic, such that one species benefits from the other's loss. Classic examples include predators and their prey, and pathogens and their hosts. Antagonistic coevolution is commonly thought to result in rapid co-evolutionary arms races \cite{dawkins1979arms}. 
When participants in a coevolutionary arms race can tamper with their opponents' control systems, as microbial pathogens do with host immune regulation \cite{bergstrom2006adaptive,schmid2003variation}, we might expect to see a series of subversion efforts and subsequent countermeasures deployed over evolutionary time. Thus one might expect rapid evolutionary divergence of the systems involved in controlling and regulating the attacks and defenses used in antagonistic interactions \cite{van1973new}.

However, antagonistic coevolution need not always lead to rapid evolutionary change. Mechanisms that prevent subversion can slow down coevolutionary arms races;  the field of cryptography abounds with examples of such systems. In a prescient 1955 letter only recently declassified, John Nash anticipated this result \cite{NashNSA}: 

\begin{quotation}...for almost all sufficiently complex types of enciphering, especially where the instructions given by different portions of the key interact complexly with each other in the determination of their ultimate effects on the enciphering, the mean key computation length increases exponentially with the length of the key....As ciphers become more sophisticated the game of cipher breaking by skilled teams, etc., should become a thing of the past.
\end{quotation}

Nash was right; one important example is the RSA cryptosystem \cite{rivest1978method}, in which two parties' communication over a network cannot be decoded by adversaries, unless they successfully find the prime factors of a large number. Prime factorization has been proven to be computationally difficult (or in the parlance of computer science, has high time complexity) and so the system is effectively secure. The main insight for RSA was that a mechanism can be made secure against subversion by using intractability or complexity as a defense. In this paper we explore how {\em defensive complexity} strategies can be generalized to domains beyond cryptography --- for example, immunology. 

We review evidence showing that the overarching control logic of the adaptive immune system (as revealed in the dynamics and structure of the immune system) is phylogenetically conserved in mammalian lineages. By contrast, other facets of the immune system show evidence of selection. Thus both the phylogenetic evidence and the structure of immune control logic are consistent with the defensive complexity hypothesis. 

To model the effect of defensive complexity on antagonistic coevolution, we introduce a new evolutionary game, the {\em signal tampering game}. This game features two players, the {\em controller} and the {\em adversary}. The controller aims to respond appropriately to the state of the environment. 
To do this, the controller deploys a control system intermediating between sensors that receive a cue about the state of the world, and effectors that take an action.  This control system is codified as a control logic with the cue as input and the effector responses as outputs. The controller's payoff is a function of the world state and the actions taken. The adversary aims to interfere, and can do so by tampering with some of the signals in the control logic.
 
We study what happens when this game is played in an evolutionary context. We pay particular attention to the case in which the controller must first deploy a control logic, and the adversary then has multiple periods in which to learn how to manipulate it. Such a state of affairs could come about for many reasons. One common biological scenario is when learning occurs at the population level by the mechanism of evolution by natural selection and evolutionary rates differ, as is the case in a vertebrate host deploying an immune control system against rapidly evolving pathogen adversaries \cite{bergstrom2003red,damore2011slowly}.  This scenario provides us a well defined and formally tractable learning system to study, in which an asexual adversary population evolves on a fitness landscape and the relevant phenotype space corresponds to the set of possible manipulations to the control system. Questions of which control systems are hard to learn reduce to questions about the rate of evolution given the fitness landscape induced for the adversary by the control system. We characterize rates of evolution for crossing long fitness valleys and demonstrate that even for a large adversary population with a high mutation rate, long fitness valleys are hard to cross. This is counter-intuitive when applied to pathogen populations, because one might think that their large population size and high mutation rate would allow them to conduct a parallel search of many paths through fitness space. In practice this argument turns out to be insufficient because the effective population size is not large enough to populate all possible pathways to the global optimum.   

We show that complex control networks can generate sign epistasis in this fitness landscape, thereby sculpting fitness valleys that must be crossed and reducing the rate at which an antagonist population evolves to subvert the network. Thus control systems afford defensive complexity against natural adversaries if they induce fitness functions for attempted manipulation that take the form of adaptive landscapes with long deep fitness valleys. Where sufficient defensive complexity is in place, antagonistic co-evolution can lead to long periods of structural conservation instead of rapid change driven by an ongoing arms race.  

\section{Results}

\subsection*{Phylogenetic conservation of immune control}

Rapid progress in molecular and cellular biology has provided a detailed picture of the wiring diagram of interacting cells and molecules in the immune systems of mice and men.  Even a cursory glance at textbooks of immunology \cite{Kuby} and review papers reveals the extraordinary complexity of this wiring diagram.  
Quantitative measures of complexity likewise indicate that the immune system is an outlier among physiological systems, with denser and more rich network motifs than other biological signaling networks \cite{Alon}. 
The immune signaling and control network in mice is remarkably similar to that in men; hence the utility of a mouse model system for understanding human immune function.  We now discuss how advances in sequencing and the development of phylogenetic tools for the analysis of these sequences gives us the ability to look at the evolution of immune genes in unprecedented detail.

Sequencing of the genomes of various mammals including mice, primates and humans reveals that the genes associated with some immune functions have evolved faster than the rest of the genome \cite{mousegenome,Bustamante1, Bustamante2, Bustamante3}.  This has been done in a number of ways.  The mouse genome consortium measured the rates of non-synonymous and synonymous mutations in different genes in the rodent and humans genomes, and used this to identify rapidly evolving genes  \cite{mousegenome}.  A similar study scanned for positively selected genes in the genomes of humans and chimpanzees \cite{Bustamante2} as well as other mammals \cite{Bustamante3}.  A more refined approach to identifying selection based on the McDonald-Kreitman test compared the substitution rates of non-synonymous and synonymous changes between the human and chimpanzee genome with the amount of polymorphism of each type observed in the human population \cite{Bustamante1}.   

While many rapidly evolving genes are involved in immune function, not all immune system genes are rapidly evolving. Cytokine signaling networks associated with T-helper differentiation (such as Th1 and Th2) are conserved. For example, the cytokine Interferon-$\gamma$ (IFN-$\gamma$) shows evidence of conservation both in structure and function among vertebrates \cite{Savan09} and is subject to strong purifying selection \cite{manry2011evolutionary}. Among murine and human lineages all but one of the type I cytokines (Interleukin 2,3,4,5,6,7,9,11,12,13,15, etc.) and cytokine receptors have homologs and are thus strongly conserved \cite{Boulay03}. Most mammalian lineages, including murine and primate lineages, show phylogenetic evidence for purifying selection  for Interleukin 2 (IL-2) \cite{zelus2000fast}. The cytokine Interleukin 12 (IL-12) shows evidence of functional conservation between chickens and higher vertebrates \cite{Degen04}. The cytokine Interleukin 16 (IL-16) shows evidence of both structural and functional conservation between human and murine lineages \cite{Keaneetal98}. The cytokines Interleukin 1 and 2 (IL-1 and IL-2) also show evidence of conservation \cite{Cohen92}. These cytokines are all implicated in the Th1/Th2 control logic of the adaptive immune response. In addition, mucosal-associated invariant T-cell subsets (MAIT) and Natural Killer T-Cell (NKT) subsets are conserved between murine and human lineages \cite{Treiner05,Brossay98}. Finally, none of the above cytokines or cytokine receptors were highlighted by recent screens of the mouse, primate, and mammalian genomes as having fast rates of evolution or evidence of positive selection \cite{mousegenome,Bustamante1,Bustamante2,Bustamante3}. Specifically, among the PANTHER ontology categories highlighted as showing positive selection, cytokine and cytokine receptor categories were not included \cite{Bustamante1,Bustamante2,Bustamante3}. 

In contrast a number of genes including those at the major histocompatibility (MHC) loci, the antibody complex, and NK cell recognition via the Ly49 receptor belong to the set of rapidly evolving genes. 

The more rapid evolution of these immune function genes is likely to be due to host-pathogen interactions.  These interactions could give rise to positive selection as a consequence of rapid host-pathogen coevolution \cite{Bustamante2} or as a consequence of diversifying selection at these loci.    

What can we learn from the conservation of immune cytokine control logic and divergence of other immune structures?  

The rapid evolution of immune genes such as those associated the MHC and antibodies has been explained by pathogen-host co-evolution or selection for diversification, but this does not explain why the genes associated with immune control logic are relatively conserved.  One can envisage a number of hypotheses for the relative conservation of immune control logic. 

One possibility is that the complexity of the immune control logic results in fragility (in a Rube Goldberg sense) and this hinders further adaptive evolution of that system \cite{fraser2002evolutionary}. Another way to describe this is that the immune control logic is at a local maximum on the fitness landscape. However, we would not expect a Rube Goldberg mechanism to be broadly robust against pathogens that evolve to disrupt the immune response. The effectiveness of the immune system to deal with most pathogens allows us to reject this hypothesis.  

Another possibility is that there simply hasn't been enough time for the vertebrate lineages to diverge in their control logics. The apparent phylogenetic conservation could be merely an artifact of the timescale at which the pathogens and immune system co-evolve. However, the too-short-timescale hypothesis is only consistent with phylogenetic conservation of immune control logic, not with phylogenetic divergence of other immune structures such as the MHC. The relatively rapid pace of change at loci such as the MHC suggests that this is not the solution.

Finally, we consider the hypothesis that the control logic of the adaptive immune system has the property of defensive complexity. The defensive complexity hypothesis is consistent with both the conservation of the immune control logic and more rapid evolution in other immune system-related genes. The conservation of immune control logic follows from the slowing down of the coevolutionary arms race. Divergence of other immune structures across lineages is consistent with selection for diversification or the presence of an arms race that isn't impeded by defensive complexity. In order to model defensive complexity in signaling networks like the cytokine signaling network, we now turn to game theory and population genetics.         

\subsection*{Signal-tampering games}

To exploit a signaling system, an adversary must (1) construct or disrupt signals used in the system, and (2) do so in a way that increases its own fitness. In host-pathogen interactions, step 1 is often simple. For example, viruses readily perturb the cytokine signaling network used by the host, by altering gene expression or by producing cytokine mimics and antagonists 
\cite{LindaGooding,Palese:2005:book}. The latter problem---manipulating the signals in advantageous ways---may be much harder. This is the challenge we focus on here. To do so, we make the {\em universal construction assumption}: the adversary can construct any signal, but does not know what the signals do. Making the universal construction assumption allows us to reformulate the problem of evolving to manipulate the host's control network as a learning problem. Defensive complexity then reduces to non-learnability of the control system by the adversary population.  

We consider the two-player game between controller and adversary illustrated in figure 1A. In each instance of the game, a cue contains information about the state of the environment. The controller aims to transduce this cue into an appropriate response. To do so, the controller's sensory apparatus detects the cue, and produces internal signals that will trigger the controller's response. The response is determined by a control logic tuned to accomplish some task $T$. This control logic is selected from the set $L(T)$ of minimal-cost control logics for this task, i.e., from a set of control logics that all perform optimally on the task $T$. The adversary aims to alter the controller's response and does so by perturbing the controller's internal signals.

\begin{figure*} 
\label{fig1}
\centering
\includegraphics[scale=.45]{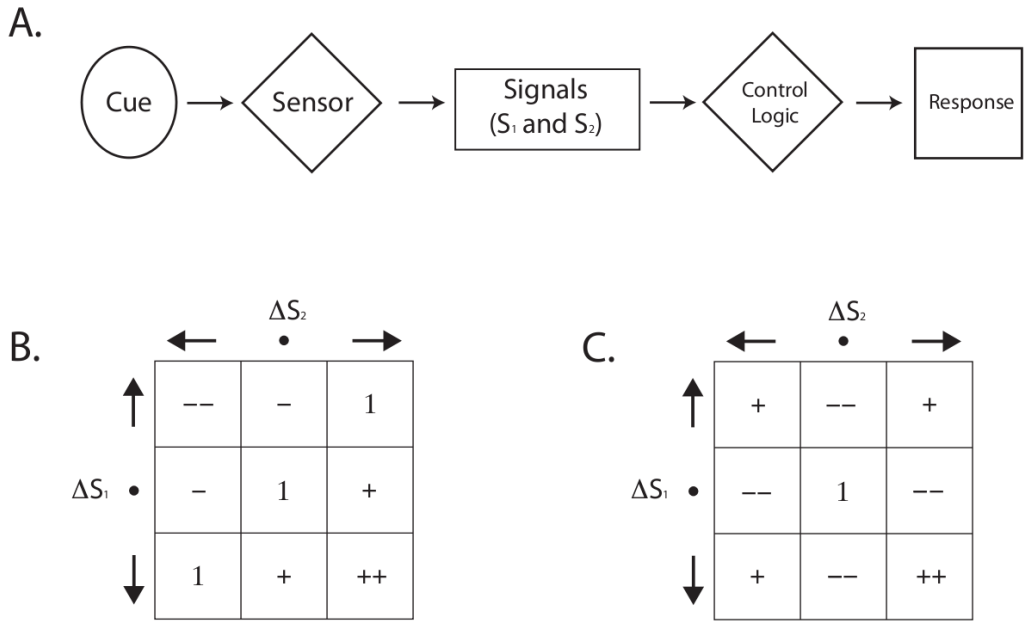}
\caption{{\bf Signal-tampering game}. A. The controller and adversary play the following game: the controller deploys a control logic and the adversary tries to subvert it. The controller has a sensor that transduces the cue into signals $S_1$ and $S_2$. The control logic then processes these signals to determine the response. The payoff to the controller is maximized when the signals are unperturbed, whereas the fitness of the adversary is some arbitrary function of the responses. Note that this game is thus generically non-zero-sum. The adversary cannot control either the cue or the response directly, but can manipulate the response by tampering with the signals $S_1$ and $S_2$. In our game, the adversary can upregulate, downregulate, or leave each signal unchanged. B. A simple control logic induces a single-peaked fitness landscape on which the adversary can easily evolve to the global optimum. C. A more complex control logic can induce a fitness landscape with deleterious fitness valleys that prevents the adversary from rapidly evolving to the global optimum.}
\end{figure*}

We explore a game in which the control logic operates on two signals $S_1$ and $S_2$. We model the control logic as a simple branching logic on the signals as described in the Methods. 
 A fitness landscape is induced by the example control network as follows: the adversaries attempt tamper by up- or down-regulating the signals $S_1$ or $S_2$.  Up-regulation or down-regulation of the signals cause changes to the response value $R$. For modeling purposes we only need to describe the control logic and its fitness consequences on the adversary as a function of perturbations to the signals. We can do so in a way that is equivalent to using a particular family of control logics: in the Methods we show the correspondence between control logics from this family and control logics with inputs  described in terms of perturbations to the signals. 
 
We outline two simple examples of control logics and the fitness landscapes that they induce; detailed specifications are provided in the methods. Let the fitness of the adversary be 1 when there is no perturbation to the control logic. First consider a simple signaling network that has the control logic $L_s$ and induces a fitness landscape for the adversary as given by Figure 1b. This fitness landscape is a slope, with each step toward down-regulating $S_1$ and up-regulating $S_2$ being progressively more beneficial for the adversary. Second, contrast this with a control logic $L_c$ that is more complex in that it requires more logical operations per conditional IF instruction. The control logic $L_c$ generates a fitness landscape with fitness valleys as shown in Figure 1c.  On this landscape, any perturbation of a single signal away from the starting state will decrease the adversary's fitness. We show that for a generalized version of the more complex control logic, the learning time is exponential in the number of signals. 

We consider a control logic to be effectively unlearnable by natural selection if the learning time for this logic is exponential in the number of signals $n$. In this case, the controller can force the learning time to blow up exponentially by adding even a modest number of signals. The major result of this section is that one can construct an effectively unlearnable control logic, analogous to $L_{c}$, that is reasonably simple (i.e., has a number of elementary operations linear in $n$, see Appendix for more details). Formally, 
\begin{theorem}
There exists a reasonably simple $n$-signal control logic that requires a number of generations exponential in $n$ to learn. 
\end{theorem}

We prove this theorem in the Appendix; the basic intuition is as follows. Think of the $n$-dimensional $3 \times 3 \times 3 \times ... \times 3$ hypercube where each dimension represents perturbations (down, none, up) to one signal. We establish a control logic by which  all corners of the hypercube are global optima and the center is a local optimum. All other spaces are fitness valleys. In other words, we construct a control logic in which global maxima occur only where each and every signal has been altered from its default value. To reach a global optimum, an adversary needs to tamper with $n$ different signals. Then we show that from the starting place where signals are left unperturbed, the expected number of generations needed to produce one of these beneficial $n$-mutants is exponential in $n$. 

This construction is just one simple example. More complicated fitness landscapes could lead pathogen populations on detours through a sequence of local maxima, delaying convergence to the global maximum.

\section{Discussion}

In his original letter to the NSA, John Nash conjectured that one could construct a cryptographic system with $n$ complexly interacting parts would take time exponential in $n$ to subvert. Nash felt this conjecture was of critical importance for cryptography, because by judicial use of defensive complexity encoders will be able to slow the arms race between with decoders. Here we have argued that Nash's conjecture is important for biological scenarios of antagonistic coevolution as well, because defensive complexity can slow the coevolutionary arms race between pathogens and the adaptive immune systems of vertebrates of their vertebrate hosts.

Our framework suggests that the kinds of signaling networks present in the immune system induce a complex fitness landscape with valleys and local maxima for pathogens attempting to deceive the immune system. Our results are consistent with two broad observations pertaining to immunology. First, we have cited evidence in mammalian lineages of phylogenetic conservation for immune control logic and positive selection for other aspects of immune functions. This relative conservation is consistent with the predictions of the defensive complexity model;  a sufficiently complex immune circuitry would limit the extent to which rapidly evolving viral and bacterial pathogens provoke coevolutionary arms races that in turn would drive divergence between the immune systems of mice and men.
Second, immunologists have found it considerably difficult to decipher the rules behind the functioning of the immune system.  Defensive complexity might be expected to involve complex rules which are difficult for pathogens to exploit and immunologists to understand. 

\section{Materials and Methods}

We implement the control logics for our example networks using Kleene's three-valued logic; for truth tables in Kleene's logic, see ref. \cite{kleene1952introduction}. We denote the perturbations to signal $i$ with $\Delta S_i$. Each perturbation can take on values  $+$, $-$ or $\bullet$. When instantiating the control logic, we interpret $+$ as $1$, $-$ as $0$, and $\bullet$ as an input being absent or unknown. For example, in the NOT ($\neg$) operation, when given $\bullet$ as input, the output is also $\bullet$. The AND operator will evaluate to True only if both its inputs are True. The OR operator will evaluate to True only if at least one of its inputs is True. After the branching program terminates---which happens when a particular conditional evaluates to True---then the rest of the conditionals cannot be triggered, and so the order of specification for the conditionals matters a great deal. 

 We use the notation $\land$ for the AND logic gate, we use $\lor$ for the OR logic gate, $\oplus$ for XOR, and we use $\neg$ for NOT. If $m$ is the fitness of the adversary then we can implement the control logic $L_s$ as follows: 

\begin{enumerate}
\item IF $ \Delta S_1 \land \neg \Delta S_2$, $\Delta R = +2m$
\item IF $ \Delta S_1 \oplus \neg \Delta S_2$, $\Delta R = +m$
\item IF $ \Delta S_2 \land \neg \Delta S_1$, $\Delta R=-2m$
\item IF $ \Delta S_2 \oplus \neg \Delta S_1$, $\Delta R=-m$\item ELSE $\Delta R=+0$
\end{enumerate}

We can implement control logic $L_c$ as follows:
\begin{enumerate}
\item IF $\Delta S_2 \land \neg \Delta S_1$, $\Delta R = -2m$
\item IF $ (\Delta S_1 \land \Delta S_2) \lor ( \Delta S_1 \land  \neg \Delta S_2) \lor (\neg \Delta S_1 \land \neg \Delta S_2)$, $\Delta R =-m$
\item IF $ (\Delta S_1 \oplus \neg \Delta S_2) \lor (S_2 \oplus \neg \Delta S_1)$, $\Delta R = +2m$
\item ELSE $\Delta R = +0$
\end{enumerate}

We have thus far discussed control logics in terms of perturbations to signals, but this is only a way of simplifying the full control logic for the purpose of discussion. In fact, we can construct control logics equivalent to the perturbation-based control logics based on the following substitutions. We say that $\bullet$ corresponds to the default value $d_{i}(c)$ that is set by the control logic to signal $i$ based on the value of the cue $c$ (more explicitly, by the sensor, based on the value of $c$), with $+$ and $-$ corresponding to any quantity larger than $d_{i}(c)$ and smaller than $d_{i}(c)$ respectively. Consider $h_{i}$ and $l_{i}$, which satisfy the inequality $h_{i} > d_{i}(c) > l_{i}$ for all $i$. We can translate each gate over perturbations into a gate over the signals and the original cue. For $\neg \Delta S_i$ we have 

$$
g_{\neg}(S_{i},c) =
\begin{cases}
l_{i}, & \text{if }S_{i} > d_{i}(c) \\
d_{i}(c), & \text{if } S_{i} = d_{i}(c)\\
h_{i}, & \text{if } S_{i} < d_{i}(c)
\end{cases}
$$

For $\Delta S_{1} \land \Delta S_{2}$ (with a ternary AND gate), we have $g_{\land}(S_{1},S_{2},c) = (S_{1} > d_{1}(c)) \land  (S_{2} > d_{2}(c))$ where $\land$ is the ordinary boolean AND gate. For $\Delta S_{1} \lor \Delta S_{2}$, we have $g_{\lor}(S_{1},S_{2},c) = S_{1} + S_{2} >  d_{1}(c) + d_{2}(c)$. For $\Delta S_{1} \oplus \Delta S_{2}$, we have $g_{\oplus}(S_{1},S_{2},c) = g_{\land}(S_{1},g_{\neg}(S_{2},c),c) \lor g_{\land}(g_{\neg}(S_{1},c),S_{2},c)$ where $\lor$ is the ordinary boolean OR gate.

\section{Acknowledgements}
This work was funded in part by the MIDAS Center for Communicable Disease Dynamics NIH 5U54GM088558-02 and by NSF grant EF-1038590 (CTB), and NIH R01 AI049334 and MIDAS U01-GM070749 (RA). The authors also wish to thank Daril Vilhena, Philip Johnson and Frazer Meacham for helpful comments. 

\bibliographystyle{unsrt}

\begin{thebibliography}{10}

\bibitem{dawkins1979arms}
R.~Dawkins and J.R. Krebs.
\newblock Arms races between and within species.
\newblock {\em Proceedings of the Royal Society of London. Series B. Biological
  Sciences}, 205(1161):489--511, 1979.

\bibitem{bergstrom2006adaptive}
C.T. Bergstrom and R.~Antia.
\newblock How do adaptive immune systems control pathogens while avoiding
  autoimmunity?
\newblock {\em Trends in ecology \& evolution}, 21(1):22--28, 2006.

\bibitem{schmid2003variation}
P.~Schmid-Hempel.
\newblock Variation in immune defence as a question of evolutionary ecology.
\newblock {\em Proceedings of the Royal Society of London. Series B: Biological
  Sciences}, 270(1513):357, 2003.

\bibitem{van1973new}
L.~Van~Valen.
\newblock A new evolutionary law.
\newblock {\em Evolutionary theory}, 1(1):1--30, 1973.

\bibitem{NashNSA}
J.~Nash.
\newblock Letter to the national security agency.
\newblock Declassified recently by the NSA, 1955.

\bibitem{rivest1978method}
R.L. Rivest, A.~Shamir, and L.~Adleman.
\newblock A method for obtaining digital signatures and public-key
  cryptosystems.
\newblock {\em Communications of the ACM}, 21(2):120--126, 1978.

\bibitem{bergstrom2003red}
C.T. Bergstrom and M.~Lachmann.
\newblock The red king effect: when the slowest runner wins the coevolutionary
  race.
\newblock {\em Proceedings of the National Academy of Sciences of the United
  States of America}, 100(2):593, 2003.

\bibitem{damore2011slowly}
J.A. Damore and J.~Gore.
\newblock A slowly evolving host moves first in symbiotic interactions.
\newblock {\em Evolution}, 2011.

\bibitem{Kuby}
T.J. Kindt, R.A. Goldsby, B.A. Osborne, and J.~Kuby.
\newblock {\em Kuby immunology}.
\newblock WH Freeman, 2007.

\bibitem{Alon}
Z.~Frankenstein, U.~Alon, I.R. Cohen, et~al.
\newblock The immune-body cytokine network defines a social architecture of
  cell interactions.
\newblock {\em Biol Direct}, 1(32):1--15, 2006.

\bibitem{mousegenome}
A.T. Chinwalla, L.L. Cook, K.D. Delehaunty, G.A. Fewell, L.A. Fulton, R.S.
  Fulton, T.A. Graves, L.D.W. Hillier, E.R. Mardis, J.D. McPherson, et~al.
\newblock Initial sequencing and comparative analysis of the mouse genome.
\newblock {\em Nature}, 420(6915):520--562, 2002.

\bibitem{Bustamante1}
C.D. Bustamante, A.~Fledel-Alon, S.~Williamson, R.~Nielsen, M.T. Hubisz,
  S.~Glanowski, D.M. Tanenbaum, T.J. White, J.J. Sninsky, R.D. Hernandez,
  et~al.
\newblock Natural selection on protein-coding genes in the human genome.
\newblock {\em Nature}, 437(7062):1153--1157, 2005.

\bibitem{Bustamante2}
R.~Nielsen, C.~Bustamante, A.G. Clark, S.~Glanowski, T.B. Sackton, M.J. Hubisz,
  A.~Fledel-Alon, D.M. Tanenbaum, D.~Civello, T.J. White, et~al.
\newblock A scan for positively selected genes in the genomes of humans and
  chimpanzees.
\newblock {\em PLoS biology}, 3(6):e170, 2005.

\bibitem{Bustamante3}
C.~Kosiol, T.~Vinar, R.R. Da~Fonseca, M.J. Hubisz, C.D. Bustamante, R.~Nielsen,
  and A.~Siepel.
\newblock Patterns of positive selection in six mammalian genomes.
\newblock {\em PLoS genetics}, 4(8):e1000144, 2008.

\bibitem{Savan09}
R.~Savan, S.~Ravichandran, J.R. Collins, M.~Sakai, and H.A. Young.
\newblock Structural conservation of interferon gamma among vertebrates.
\newblock {\em Cytokine \& growth factor reviews}, 20(2):115--124, 2009.

\bibitem{manry2011evolutionary}
J.~Manry, G.~Laval, E.~Patin, S.~Fornarino, M.~Tichit, C.~Bouchier, L.B.
  Barreiro, and L.~Quintana-Murci.
\newblock Evolutionary genetics evidence of an essential, nonredundant role of
  the ifn-$\gamma$ pathway in protective immunity.
\newblock {\em Human mutation}, 2011.

\bibitem{Boulay03}
J.L. Boulay, J.J. O'Shea, and W.E. Paul.
\newblock Molecular phylogeny within type i cytokines and their cognate
  receptors.
\newblock {\em Immunity}, 19(2):159--163, 2003.

\bibitem{zelus2000fast}
D.~Zelus, M.~Robinson-Rechavi, M.~Delacre, C.~Auriault, and V.~Laudet.
\newblock Fast evolution of interleukin-2 in mammals and positive selection in
  ruminants.
\newblock {\em Journal of Molecular Evolution}, 51(3):234--244, 2000.

\bibitem{Degen04}
W.G.J. Degen, N.~van Daal, H.I. van Zuilekom, J.~Burnside, and V.E.J.C.
  Schijns.
\newblock Identification and molecular cloning of functional chicken il-12.
\newblock {\em The Journal of Immunology}, 172(7):4371, 2004.

\bibitem{Keaneetal98}
J.~Keane, J.~Nicoll, S.~Kim, D.M.H. Wu, W.W. Cruikshank, W.~Brazer, B.~Natke,
  Y.~Zhang, D.M. Center, and H.~Kornfeld.
\newblock Conservation of structure and function between human and murine
  il-16.
\newblock {\em The Journal of Immunology}, 160(12):5945, 1998.

\bibitem{Cohen92}
N.~Cohen and L.~Haynes.
\newblock The phylogenetic conservation of cytokines.
\newblock {\em Phylogenesis of immune functions}, pages 242--68, 1992.

\bibitem{Treiner05}
E.~Treiner, L.~Duban, I.C. Moura, T.~Hansen, S.~Gilfillan, and O.~Lantz.
\newblock Mucosal-associated invariant t (mait) cells: an evolutionarily
  conserved t cell subset.
\newblock {\em Microbes and infection}, 7(3):552--559, 2005.

\bibitem{Brossay98}
L.~Brossay, M.~Chioda, N.~Burdin, Y.~Koezuka, G.~Casorati, P.~Dellabona, and
  M.~Kronenberg.
\newblock Cd1d-mediated recognition of an $\alpha$-galactosylceramide by
  natural killer t cells is highly conserved through mammalian evolution.
\newblock {\em The Journal of experimental medicine}, 188(8):1521, 1998.

\bibitem{fraser2002evolutionary}
H.B. Fraser, A.E. Hirsh, L.M. Steinmetz, C.~Scharfe, and M.W. Feldman.
\newblock Evolutionary rate in the protein interaction network.
\newblock {\em Science}, 296(5568):750, 2002.

\bibitem{LindaGooding}
L.R. Gooding.
\newblock Virus proteins that counteract host immune defenses.
\newblock {\em Cell}, 71(1):5, 1992.

\bibitem{Palese:2005:book}
Peter Palese, editor.
\newblock {\em Modulation of host gene expression and innate immunity by
  viruses}.
\newblock Springer, 2005.

\bibitem{kleene1952introduction}
S.C. Kleene, NG~de~Bruijn, J.~de~Groot, and A.C. Zaanen.
\newblock Introduction to metamathematics.
\newblock 1952.

\bibitem{shannon1949synthesis}
C.E. Shannon.
\newblock The synthesis of two-terminal switching circuits.
\newblock {\em Bell System Technical Journal}, 28(1):59--98, 1949.

\bibitem{Fisher22}
R.A. Fisher et~al.
\newblock On the dominance ratio.
\newblock {\em Proceedings of the Royal Society of Edinburgh}, 42:321--341,
  1922.

\bibitem{Haldane29}
J.B.S. Haldane.
\newblock A mathematical theory of natural and artificial selection, part v:
  selection and mutation.
\newblock In {\em Mathematical Proceedings of the Cambridge Philosophical
  Society}, volume~23, pages 838--844. Cambridge Univ Press, 1927.

\bibitem{Campbell03}
RB~Campbell.
\newblock A logistic branching process for population genetics.
\newblock {\em Journal of theoretical biology}, 225(2):195--203, 2003.

\bibitem{GaoHanSchilling11}
F.~Gao, L.~Han, and K.~Schilling.
\newblock On the rate of convergence of iterated exponentials.
\newblock {\em Journal of Applied Mathematics and Computing}, pages 1--8, 2011.

\bibitem{FarringtonGrant}
CP~Farrington and AD~Grant.
\newblock The distribution of time to extinction in subcritical branching
  processes: applications to outbreaks of infectious disease.
\newblock {\em Journal of applied probability}, 36(3):771--779, 1999.

\end{thebibliography}

\begin{appendix}

\section{Proof of main theorem}

Before proving the main theorem, we first specify precisely what we mean by a ``reasonably simple'' control logic. A control logic can be described as a set $F$ of functions---one for each conditional in the branching logic. Each of these functions $f$ in $F$ takes some input and returns either a 1 (if the formula specifying the conditional evaluates to True) or 0 (otherwise, see Methods). In the example of section 2, the control logic would take the signals $S_1$ and $S_2$ as inputs, and determine the appropriate response by evaluating each conditional. We can quantify circuit complexity as follows. The circuit complexity of branch $f$ of the control logic is simply the minimum number of ternary logic gates (see Methods) needed to implement $f$ \cite{shannon1949synthesis}. The circuit complexity of the full logic $F$ is the maximum circuit complexity over $f$ in $F$. The circuit complexity of a reasonably simple control logic is $O(n)$. 
 
To cross the valley in $k$ dimensions, the individual needs to found a lineage that accumulates at least $k$ successive mutations and thus this lineage must survive in the fitness valley for at least $k$ generations. Alternatively, a lineage could pick up multiple mutations in each generation---but this won't help. Consider the case that $w$ mutants arise each time step, and the lineage must survive for at least $k/w$ generations. When $k$ grows very large and $w$ is constant, asymptotic analysis tells us that the individual needs to found a lineage that survives for at least $\Theta(k)$ generations (because $\lim_{k \rightarrow \infty} |\frac{k/w}{k}| = (1/w)$ and $\lim_{k \rightarrow \infty} |\frac{k}{k/w}| = w$). In this case, generating $w$-mutants each time-step doesn't speed up crossing the fitness valley appreciably. But what $w$ is allowed to grow with $k$? For example, let $w = ck$, where $0 < c \leq 1$. In that case we would only require a lineage to survive for $1/c$ steps, and this is independent of $k$. But the time that it would take to generate the $w$ mutant required for each successive step is on average $\frac{1}{\mu^{ck}} = \left(\frac{1}{\mu} \right)^{ck}$ where $\mu$ is the mutation rate. There is no escape from an average time exponential in $k$.

It is therefore sufficient to determine how likely an individual who steps down into the valley is to found a lineage that survives in the valley for at least $k$ generations, as $k$ grows very large. If this probability decreases exponentially with $k$, the average time until we get the first individual who is destined to succeed will increase exponentially in $k$.  To model the fate of such a lineage, we approximate it as a subcritical branching process \cite{Fisher22,Haldane29,Campbell03}. If the relative fitness of an individual in the valley is $0<\lambda<1$, this individual will have a Poisson number of offspring with mean $\lambda$. Each of these offspring will themselves have Poisson numbers of offspring, again with mean $\lambda $. First we prove a technical lemma, and then we show that the average time until we get the first individual who is destined to succeed increases exponentially in $k$ as $k$ grows large.

\begin{lemma} \label{powertower}
$1-e^{-\lambda} f[e^{\lambda e^{-\lambda}},k] = |e^{-\lambda}-1|\lambda^{O(k)}$, where $ \lambda < 1$ and $f[b,k]$ is the tetration function $b^{b^{...^{b}}}$ with $k$ levels of iterated exponentiation.
\end{lemma}
\begin{proof}
First we use the result (Theorem 2 in ref.~\cite{GaoHanSchilling11}) that since $\lambda < 1$, $f[e^{\lambda e^{-\lambda}},k]$ converges to $e^{\lambda}$ at a linear rate $\lambda$. Thus for any $n$, we have $|f[e^{\lambda e^{-\lambda}},n+1] - e^{\lambda}| =O( \lambda |f[e^{\lambda e^{-\lambda}},n] - e^{\lambda}|)$.

From the above, terminating with $k$ and solving the recurrence relation, we obtain: 
$$|f[e^{\lambda e^{-\lambda}},k] - e^{\lambda}| = \lambda^{O(k)} |f[e^{\lambda e^{-\lambda}},0] - e^{\lambda}|,$$
and so by the definition of tetration,  $|f[e^{\lambda e^{-\lambda}},k] - e^{\lambda}| =  \lambda^{O(k)} |1 - e^{\lambda}|$. The symmetry property of absolute value implies that $|e^{\lambda} - f[ e^{\lambda e^{-\lambda}},k]| =  \lambda^{O(k)} |1 - e^{\lambda}|$

Multiply both sides by $e^{-\lambda}$ to obtain:
$$|1-e^{-\lambda} f[e^{\lambda e^{-\lambda}},k]| =  \lambda^{O(k)} |e^{-\lambda}-1|$$
since for $c>0$, $c|a-b| = |c(a-b)|$. Because $f_k = e^{-\lambda} f[e^{\lambda e^{-\lambda}},k]$ is the cumulative distribution function (CDF) of some distribution \cite{FarringtonGrant}, this implies that $0\leq f_k \leq 1$. Therefore, $0 \leq 1-f_k \leq 1$, and in particular, $1-f_k \geq 0$. Consequently we can simplify the absolute value as follows: 
$$1-e^{-\lambda} f[e^{\lambda e^{-\lambda}},k] = \lambda^{O(k)} |e^{-\lambda}-1|$$
And thus this establishes the lemma. \qed
\end{proof}

\begin{lemma} \label{timebound}
As $k$ grows large, the average time until we get the first individual that is destined to succeed assuming its lineage is modeled by a subcritical Poisson branching process (with average number of offspring $0<\lambda < 1$) is at least $(1/\lambda)^{\Omega(k-1)} / (N\,\mu)$ generations on average.
\end{lemma}
\begin{proof}
By the subcriticality of the branching process that generates the lineage, $\lambda < 1$. The random variable $Y$ is the number of generations the lineage generated by the branching process survives. The CDF of the branching process $f_{n} = P(Y \leq n)$ gives us the probability of a lineage surviving no more than $n$ generations. Therefore, for our purposes, we need to characterize the probability of non-extinction for $k$ generations, which means we must characterize $P(Y > k-1) = 1-f_{k-1}$.   

By equation 2.5 in Farrington and Grant \cite{FarringtonGrant}, when $\lambda \leq 1$, $f_{n} = e^{-\lambda} f[e^{\lambda e^{-\lambda}},n]$, where $f[b,k]$ is the tetration function. Therefore, by Lemma \ref{powertower}, $P(Y > k-1) = |e^{-\lambda}-1|\lambda^{O(k-1)}$. The number of mutants who must enter the valley before one does so successfully is on average $1/P(Y > k-1)$, which is thus at least $(1/\lambda)^{\Omega(k-1)}$ since $0<|e^{-\lambda} - 1|<1$. Since the rate at which mutants are produced from the wild type each generation is upper bounded by $N \, \mu$, the result follows. \qed    
\end{proof}

Now to prove the main Theorem, we combine Lemma \ref{timebound} with a suitable control logic: 
 
 \begin{theorem}
There exists a reasonably simple $n$-signal control logic that requires a number of generations exponential in $n$ to learn. 
\end{theorem}

\begin{proof}
By our definition above of ``reasonably simple'', we aim to show that there exists a control logic on $n$ signals with circuit complexity $O(n)$ that is learnable in a number of generations exponential in $n$. 

Without loss of generality we will assume that the controller response with deleterious consequences to the adversaries is $R$. What this means is that the fitness of the adversary is $1-\Delta R $. 

The adversary has three possible actions, upregulate, downregulate, or do nothing. We will now build a control logic that downregulates $R$ only if the adversary upregulates or downregulates each and every one of the signals (rather than doing nothing), and otherwise upregulates $R$. We will then show that for small enough adversary populations, generating an $n$-mutant takes time exponential in $n$. 

The control logic $L$ we consider is simply:
\begin{enumerate}
\item IF $D(\Delta x_1)  \land ... \land D(\Delta x_{n})$, $ \Delta R = +0$
\item IF $D(\Delta x_1)  \lor ... \lor D(\Delta x_{n})$, $ \Delta R = \delta$\item ELSE $\Delta R = -s$
\end{enumerate}
where $D(x) = (x \Leftrightarrow \bullet)$ and $0 < \delta < 1$. Here $\Leftrightarrow$ represents the logical equivalence operator in Kleene's logic. Generating an optimal type requires a beneficial $n$-mutant with deleterious intermediates, and by Lemma \ref{timebound}, this takes at least on average $(1/\lambda)^{\Omega(n-1)}/(N \,\mu)$ generations, for $0<\lambda < 1$. Since $\lambda$ is the relative fitness of an individual in the fitness valley, $\lambda=1-\delta$ by line 2 of $L$. 

Now note that $L$ on the first line has $n$ clauses of the form $\Delta x \Leftrightarrow \bullet$, each of which has 1 gate, and connected by $n-1$ OR gates. In total then there are $n + n  - 1 = 2n -1$ gates on the first line of $L$ (and this is maximal). The result follows. \qed \end{proof}

\end{appendix}

\end{document}